\documentclass[11pt]{article}
\usepackage[left=1in,top=1in,right=1in,bottom=1in]{geometry}
\usepackage{times}

\usepackage{amsmath}
\usepackage{amssymb}
\usepackage{amsthm}
\usepackage{graphicx}
\usepackage{verbatim}
\usepackage{url}
\usepackage{algorithm}
\usepackage{algorithmic}
\newtheorem{definition}{Definition}
\newtheorem{candidatedefinition}{Candidate Definition}
\newtheorem{proposition}{Proposition}
\newtheorem{theorem}{Theorem}

\allowdisplaybreaks

\begin{document}

\title{Dominated Actions in Imperfect-Information Games\thanks{Accepted for presentation at the 65th IEEE Conference on Decision and Control (CDC 2026).}}

\author{
Sam Ganzfried\\
Ganzfried Research, Cornell University\\
\texttt{sam.ganzfried@gmail.com}
}

\date{\vspace{-5ex}}

\maketitle

\begin{abstract}
Dominance is a fundamental concept in game theory. In normal-form games dominated strategies can be identified in polynomial time. As a consequence, iterative removal of dominated strategies can be performed efficiently as a preprocessing step for reducing the size of a game before computing a Nash equilibrium. For imperfect-information games in extensive form, we could convert the game to normal form and then iteratively remove dominated strategies in the same way; however, this conversion may cause an exponential blowup in game size. In this paper we define and study the concept of dominated actions in imperfect-information games. Our main result is a polynomial-time algorithm for determining whether an action is strictly or weakly dominated by a distribution over other actions at the same information set. The algorithm applies to two-player perfect-recall games whenever the acting player has no further decision on any path after that information set. Players’ actions need not be publicly observable, and the algorithm can be applied iteratively. We show empirically that iterated removal of dominated actions substantially reduces the number of decisions in the two-player version of ``All In or Fold'' No-Limit Texas Hold'em poker.
\end{abstract}

\section{Introduction}
\label{se:intro}
In a normal-form game, (mixed) strategy $\sigma_i$ for player $i$ is \emph{strictly dominated} if there exists another (mixed) strategy $\sigma'_i$ such that $\sigma'_i$ performs strictly better than $\sigma_i$ regardless of the strategies used by the opponent(s): formally, if $u_i(\sigma'_i,s_{-i}) > u_i(\sigma_i,s_{-i})$ for all pure strategy profiles (vectors of pure strategies) $s_{-i} \in S_{-i}$ for the opposing agents. (Note that the requirement that this holds for all opposing pure strategy profiles is enough to ensure that it also holds for all mixed strategy profiles of the opponents as well.) Clearly it would be irrational for an agent to play a strictly dominated strategy, as another strategy will do strictly better regardless of the beliefs of what the opponents would play. Strategy $\sigma'_i$ \emph{weakly dominates} $\sigma_i$ if the inequality holds weakly (though strictly for at least one pure strategy profile): formally, $u_i(\sigma'_i,s_{-i}) \geq u_i(\sigma_i,s_{-i})$ for all $s_{-i} \in S_{-i}$ where the inequality is strict for at least one $s_{-i}$. (The condition that at least one inequality is strict is simply to rule out saying that a strategy is dominated by an identical strategy). 

It seems natural to simplify a game by eliminating strategies that are dominated from the game to reduce its size and focus analysis on a smaller game. It can easily be shown that applying an iterative process of removing one dominated strategy for one player, then removing one for another (or possibly the same) player in the reduced game, etc., will ultimately result in a smaller game that contains a Nash equilibrium from the original game. For this reason, this procedure of \emph{iterated removal of dominated strategies} is often performed as a preprocessing step to reduce the size of a game before computing a Nash equilibrium (or other desired solution concept). As it turns out, all processes of iterated removal of strictly dominated strategies produce the same reduced game regardless of the order of elimination in finite games (while this is not necessarily the case for iterated removal of weakly dominated strategies)~\cite{Gilboa90:Order,Maschler13:Game}. While iterated removal of weakly dominated strategies can sometimes reduce the number of equilibria, it can never create new equilibria, and therefore even this procedure is very useful as a preprocessing step for Nash equilibrium computation.

There is a linear-time algorithm for determining whether a (mixed) strategy $\sigma_i$ is strictly dominated by any pure strategy for player $i$~\cite{Shoham09:Multiagent}. This algorithm simply iterates over each pure strategy $s_i$ for player $i$ and tests whether it performs strictly better than $\sigma_i$ against each opposing pure strategy profile $s_{-i}$. This procedure has complexity $O(|S|)$, where $S = \times_i S_i$ is the set of joint pure strategy profiles, and so takes time linear in the size of the game. The procedure can also be easily adapted to produce an algorithm for determining whether a mixed strategy is weakly dominated by any pure strategy in linear time. Note that it is possible for a strategy to be dominated by a mixed strategy and not be dominated by any pure strategy~\cite{Shoham09:Multiagent}. In order to determine whether a (mixed) strategy is strictly dominated by a mixed strategy, while the above procedure does not work, it turns out that there exists a linear programming formulation that runs in polynomial time, and there also exists a linear programming formulation that determines whether a (mixed) strategy is weakly dominated by a mixed strategy~\cite{Conitzer05:Complexity_of_iterated_dominance}. Therefore, regardless of whether the strategies being tested as dominated or dominating are mixed or pure, it can be checked in polynomial time.

In an extensive-form game, eliminating a dominated strategy can leave the game tree unchanged because other strategies may still use each of its actions. We therefore study when an action at a particular information set can be eliminated. Gibson~\cite{Gibson14:Regret} defined strictly dominated actions using counterfactual values and established conditions under which counterfactual regret minimization avoids iteratively strictly dominated actions. Our definition holds the acting player's strategy fixed at all other information sets. In games with perfect recall, the two definitions of strict dominance agree at any reachable information set after which the acting player has no further decision on any path. In two-player games with perfect recall, our polynomial-time algorithm applies when the acting player has no further decision on any path after that information set, although the opponent may still act. This condition holds at every decision in the two-player poker game studied in Section~\ref{se:exp}.

\section{Extensive-form games}
\label{se:efg}
While the normal form can be used to model simultaneous moves, another representation, called the \emph{extensive form}, is generally preferred when modeling settings that have sequential moves. The extensive form can also model simultaneous moves, as well as chance events and imperfect information (i.e., situations where some information is available to only some of the agents and not to others). Extensive-form games consist primarily of a game tree; each non-terminal node has an associated player (possibly \emph{chance}) that makes the decision at that node, and each terminal node has associated utilities for the players. Additionally, game states are partitioned into \emph{information sets}, where the player whose turn it is to move cannot distinguish among the states in the same information set.  Therefore, in any given information set, a player must choose actions with the same distribution at each state contained in the information set.  If no player forgets information that they previously knew, we say that the game has \emph{perfect recall}. A game has \emph{publicly observable actions} if, for every player, no information set contains nodes reached by different actions of another player.\footnote{This assumption still allows private information and unobserved chance outcomes. It holds, for example, in poker and in many signaling and bargaining games with publicly observed actions, as well as open-bid auctions and security games with observable strategic actions.} A (behavioral) \emph{strategy} for player $i,$ $\sigma_i \in \Sigma_i,$ is a function that assigns a probability distribution over all actions at each information set belonging to $i$. In theory, every extensive-form game can be converted to an equivalent normal-form game; however, there is an exponential blowup in the size of the game representation, and therefore such a conversion is undesirable. Instead, algorithms have been developed that operate on the extensive form representation directly. It turns out that the complexity of computing equilibria in extensive-form games is similar to that of normal-form games; a Nash equilibrium can be computed in polynomial time in two-player zero-sum games (with perfect recall)~\cite{Koller94:Fast}, while the problem is hard for two-player general-sum and multiplayer games. 

One algorithm for computing an equilibrium in two-player zero-sum extensive-form games with perfect recall is based on solving a linear programming formulation~\cite{Koller94:Fast}. This formulation works by modeling each sequence of actions for each player as a variable, and is often called the \emph{sequence form LP} algorithm. Note that while the number of pure strategies is exponential in the size of the game tree, the number of action sequences is linear. The method uses several matrices defined as follows. For player 1, the matrix $\mathbf{E}$ is defined where each row corresponds to an information set (including an initial row for the ``empty'' information set), and each column corresponds to an action sequence (including an initial column for the ``empty'' action sequence). In the first row of $\mathbf{E}$ the first element is 1 and all other elements are 0; subsequent rows have -1 for the entries corresponding to the action sequence leading to the root of the information set, and 1 for all actions that can be taken at the information set (and 0 otherwise). Thus $\mathbf{E}$ has dimension $(c_1+1) \times (d_1+1)$, where $c_i$ is the number of (nonempty) information sets for player $i$ and $d_i$ is the number of (nonempty) action sequences for player $i$. Matrix $\mathbf{F}$ is defined analogously for player 2. The vector $\mathbf{e}$ is defined to be a column vector of length $c_1+1$ with 1 in the first position and 0 in other entries, and vector $\mathbf{f}$ is defined with length $c_2+1$ analogously. The matrix $\mathbf{A}$ has dimension $(d_1+1)\times(d_2+1)$. Its entry $A_{ij}$ is the sum, over all terminal nodes at which the final action sequences of players~1 and~2 have indices $i$ and $j$, respectively, of player~1's payoff multiplied by the product of the chance probabilities along the path to that node. The matrix $\mathbf{B}$ of player 2's expected payoffs is defined analogously. In zero-sum games $\mathbf{B} = -\mathbf{A}.$

Given these matrices we can solve one of two linear programming problems to compute a Nash equilibrium in zero-sum extensive-form games~\cite{Koller94:Fast}. In the first formulation, $\mathbf{x}$ is player~1's realization plan and $\mathbf{q}$ is a vector of dual variables for player~2's realization-plan constraints. In the second formulation, which is the dual of the first, $\mathbf{y}$ is player~2's realization plan and $\mathbf{p}$ is a vector of dual variables for player~1's realization-plan constraints.

\[
\begin{array}{rrl} 
&\max_{\mathbf{x},\mathbf{q}}& -\mathbf{q}^T \mathbf{f} \\ 
&\mbox{s.t.}& \mathbf{x}^T (-\mathbf{A}) - \mathbf{q}^T \mathbf{F} \leq \mathbf{0} \\
& & \mathbf{x}^T \mathbf{E}^T = \mathbf{e}^T \\
& & \mathbf{x} \geq \mathbf{0}\\
\end{array} 
\]

\[
\begin{array}{rrl} 
&\min_{\mathbf{y},\mathbf{p}}& \mathbf{e}^T \mathbf{p} \\ 
&\mbox{s.t.}& -\mathbf{A} \mathbf{y} + \mathbf{E}^T \mathbf{p} \geq \mathbf{0} \\
& & -\mathbf{F} \mathbf{y} = -\mathbf{f} \\
& & \mathbf{y} \geq \mathbf{0}\\
\end{array} 
\]

\section{Dominated actions}
\label{se:dominated-actions}
In extensive-form games, we can consider analogous concepts of strict, weak, and iterated dominance of strategies as for normal-form games. However, unlike in the normal-form setting, identification of a dominated extensive-form strategy does not necessarily allow us to reduce the size of the game, since it is possible that some of the actions played by the dominated strategy are also played by non-dominated strategies. In order to obtain the computational advantage of game size reduction, we must consider a stronger concept of \emph{dominated actions}. We first present several plausible candidate definitions for dominated actions which we demonstrate to be problematic. Our first candidate definition is given as Candidate Definition~\ref{cd:strict-strong}. This definition states that action $a_i$ for player $i$ at information set $I_i$ is strictly dominated by action $b_i$ at the same information set if every leaf node succeeding $a_i$ produces a strictly smaller payoff for player $i$ than every leaf node succeeding action $b_i.$ An analogous definition for weak dominance is given in Candidate Definition~\ref{cd:weak-strong}. 

\begin{candidatedefinition}
\label{cd:strict-strong}
If for every leaf node $N^{a_i}$ that follows action $a_i$ for player $i$ at information set $I_i$ and every leaf node $N^{b_i}$ that follows action $b_i$ for player $i$ at the same information set $I_i$, $u_i \left(N^{b_i}\right) > u_i \left(N^{a_i}\right),$ then $b_i$ \emph{strictly dominates} $a_i.$
\end{candidatedefinition}

\begin{candidatedefinition}
\label{cd:weak-strong}
If for every leaf node $N^{a_i}$ that follows action $a_i$ for player $i$ at information set $I_i$ and every leaf node $N^{b_i}$ that follows action $b_i$ for player $i$ at the same information set $I_i$, $u_i \left(N^{b_i}\right) \geq u_i \left(N^{a_i}\right)$ where inequality is strict for at least one node, then $b_i$ \emph{weakly dominates} $a_i.$
\end{candidatedefinition}

The problem with these definitions is that they are too strong; it is still possible for player 1 to strictly prefer action $b_i$ to $a_i$ regardless of the strategy used by player 2 even if Candidate Definition~\ref{cd:strict-strong} does not hold. This is illustrated in the following game depicted in Figure~\ref{fi:example}. In this game, chance makes an initial move taking each of two actions with probability $\frac{1}{2}.$ Then player 1 (red) selects one of two actions at a single information set. Player 2 (blue) then takes one of two actions after observing both the moves of chance and player 1. It is clear that action 2 for player 2 at their top information set and action 1 for player 2 at their bottom information set are both strictly dominated according to Candidate Definition~\ref{cd:strict-strong}. The smaller game obtained after removing these actions is depicted in Figure~\ref{fi:example-dom}. In the smaller game, the expected utility of playing action 1 is $0.5 (-100) + 0.5 (100) = 0$ and the expected utility of playing action 2 is $0.5 (-50) + 0.5 (-50) = -50.$ Since player 2 does not take any actions, player 1 always achieves a strictly higher expected utility by playing action 1. However, action 2 is not strictly dominated according to Candidate Definition~\ref{cd:strict-strong} because in one leaf node succeeding action 1 player 1 obtains payoff $-100$ which is lower than the payoff of $-50$ at leaf nodes following action 2.

\begin{figure*}[!ht]
\centering
\includegraphics[width=0.85\textwidth,height=\textheight,keepaspectratio]{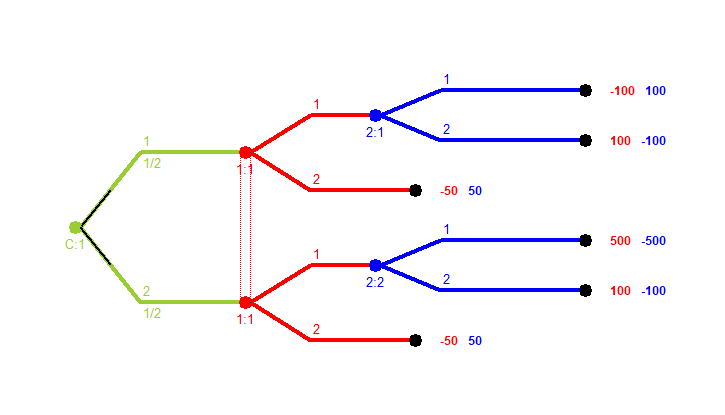}
\caption{Example two-player imperfect-information extensive-form game.}
\label{fi:example}
\end{figure*}

\begin{figure*}[!ht]
\centering
\includegraphics[width=0.85\textwidth,height=\textheight,keepaspectratio]{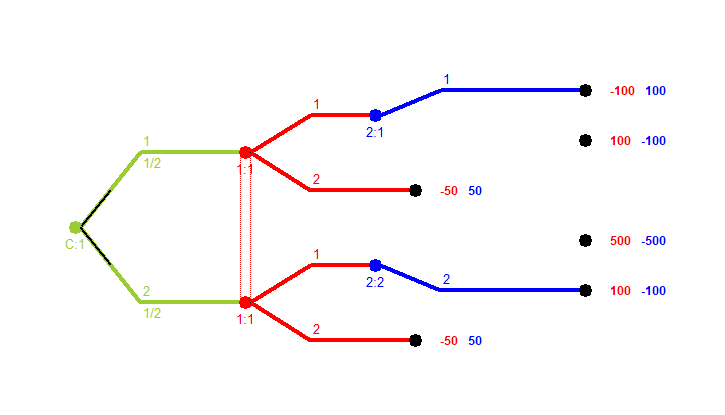}
\caption{Result of removing two dominated actions from game in Figure~\ref{fi:example}.}
\label{fi:example-dom}
\end{figure*}

Candidate Definition~\ref{cd:strict-strong} provides a sufficient condition for strict dominance, although Figure~\ref{fi:example-dom} shows that this condition is not necessary. Candidate Definition~\ref{cd:weak-strong} guarantees that switching from $a_i$ to $b_i$ cannot decrease player~$i$'s expected payoff. We refer to actions that satisfy these conditions as \emph{strongly dominated} (strictly or weakly, respectively). Both conditions can be verified efficiently by comparing the payoffs at the leaf nodes succeeding each action. Thus, it may be useful to check both conditions as a preprocessing step. In the strict case, the action can be removed immediately. In the weak case, it remains to verify that for each strategy of player~$i$ that can reach $I_i$ and plays $a_i$ there, some opposing strategy profile reaches $I_i$ with positive probability and gives player~$i$ a strict improvement from switching to $b_i$.

We next consider Candidate Definition~\ref{cd:strict}, where $u_i$ denotes the expected utility accounting for randomized moves of chance as well as potential randomization in the players' strategies. Under this candidate definition, action $a_i$ is strictly dominated at $I_i$ if there exists a behavioral strategy $\sigma^{-a_i}_i$ that never plays $a_i$ at $I_i$ (and may mix over the other actions at that information set) and that, against every opposing strategy profile, gives player~$i$ a strictly higher expected utility than every strategy that plays $a_i$ with probability one at $I_i$. Since the strategies being compared may differ at other information sets, a higher expected payoff need not result from changing the action at $I_i$ alone. A further issue is that the players may take actions earlier in the game tree that prevent the game from ever reaching $I_i$. Consider the simple example game in Figure~\ref{fi:example-offtree}. Suppose we want to apply Candidate Definition~\ref{cd:strict} to determine whether action 2 is strictly dominated for player 2. Then $\sigma^{-a_i}_i$ is the strategy for player 2 that always plays action 1, and $\sigma^{a_i}_i$ must be the strategy for player 2 that plays action 2 with probability 1. Now suppose that player 1 plays action 2 with probability 1, which we denote as $\sigma_{-i}$. Then clearly $u_i\left(\sigma^{-a_i}_i, \sigma_{-i}\right) = u_i\left(\sigma^{a_i}_i, \sigma_{-i}\right) = 0,$ since both strategy profiles will result in reaching the bottom leaf node yielding payoff 0. 

\begin{candidatedefinition}
\label{cd:strict}
Action $a_i$ for player $i$ is \emph{strictly dominated} at information set $I_i$ if there exists a behavioral strategy $\sigma^{-a_i}_i$ that plays action $a_i$ at $I_i$ with probability 0 such that for every behavioral strategy $\sigma^{a_i}_i$ for player $i$ that plays action $a_i$ with probability 1 at $I_i$, $u_i\left(\sigma^{-a_i}_i, \sigma_{-i}\right) > u_i\left(\sigma^{a_i}_i, \sigma_{-i}\right)$ for all opposing strategies $\sigma_{-i} \in \Sigma_{-i}.$
\end{candidatedefinition}

\begin{figure*}[!ht]
\centering
\includegraphics[width=0.85\textwidth,height=\textheight,keepaspectratio]{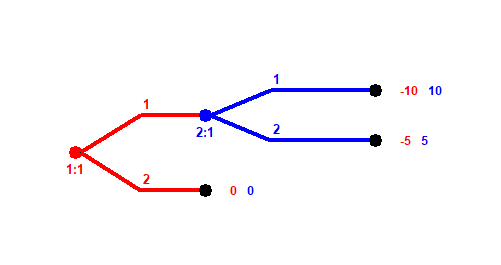}
\caption{Extensive-form game demonstrating problem with Candidate Definition~\ref{cd:strict}.}
\label{fi:example-offtree}
\end{figure*}

The problem with Candidate Definition~\ref{cd:strict} is that it allows the players to deviate from the path of play leading to the relevant information set $I_i.$ We address this limitation in our new definitions. \emph{Strictly-dominated actions} are defined in Definition~\ref{de:strict} and \emph{weakly-dominated actions} are defined in Definition~\ref{de:weak}. In these definitions, we compare payoffs only for strategy profiles that reach $I_i$ with positive probability, holding the player's strategy fixed at all other information sets. Note that these definitions apply to games with any number of players. They also consider actions that are dominated by a behavioral strategy that differs only at $I_i$ (and may mix over the other actions at $I_i$). It is easy to verify that these definitions address the issues that arose in the examples from Figure~\ref{fi:example} and Figure~\ref{fi:example-offtree}. We can apply these definitions repeatedly in succession to perform iterative removal of dominated actions, as for the normal form.

\begin{definition}
\label{de:strict}
Action $a_i$ for player $i$ is \emph{strictly dominated} at a reachable
information set $I_i$ if there exists a probability distribution $\mathbf{p}$ over
the other actions at $I_i$ with the following property. For every strategy
profile $(\sigma_i,\sigma_{-i})$ that reaches $I_i$ with positive
probability and in which $\sigma_i$ plays $a_i$ at $I_i$ with probability
one, let $\sigma_i'$ agree with $\sigma_i$ at every information set
except $I_i$, where it plays according to $\mathbf{p}$. Then
\[
u_i(\sigma_i',\sigma_{-i}) > u_i(\sigma_i,\sigma_{-i}).
\]
\end{definition}

\begin{definition}
\label{de:weak}
Action $a_i$ for player $i$ is \emph{weakly dominated} at a reachable
information set $I_i$ if there exists a probability distribution $\mathbf{p}$ over
the other actions at $I_i$ with the following property. For every strategy
profile $(\sigma_i,\sigma_{-i})$ that reaches $I_i$ with positive
probability and in which $\sigma_i$ plays $a_i$ at $I_i$ with probability
one, let $\sigma_i'$ agree with $\sigma_i$ at every information set
except $I_i$, where it plays according to $\mathbf{p}$. Then
\[
u_i(\sigma_i',\sigma_{-i}) \geq u_i(\sigma_i,\sigma_{-i}).
\]
For each such $\sigma_i$, the inequality must be strict for at least one
opposing strategy profile $\sigma_{-i}$ under which $I_i$ is reached
with positive probability.
\end{definition}

Proposition~\ref{prop:strict-dominance-ne} shows that strictly dominated actions receive probability zero at information sets reached in a Nash equilibrium, and Proposition~\ref{prop:dominance-ne-extension} shows that every Nash equilibrium of the reduced game extends to one of the original game.

\begin{proposition}
\label{prop:strict-dominance-ne}
In a perfect-recall game with any number of players, if an information set is 
reached with positive probability in a Nash equilibrium, then every strictly dominated 
action at that information set is played with probability zero.
\end{proposition}

\begin{proof}
Suppose that player $i$ plays a strictly dominated action $a_i$ at such
an information set $I_i$ with probability $\alpha>0$. Let $\mathbf{p}$ be the
distribution in Definition~\ref{de:strict}. Holding all other strategies
and all of player $i$'s choices outside $I_i$ fixed, let $U(a_i)$ and
$U(\mathbf{p})$ denote player $i$'s expected payoffs from playing $a_i$ with
probability one and from playing according to $\mathbf{p}$ at $I_i$, respectively.
Changing the choice at $I_i$ does not affect the probability of reaching
it, so Definition~\ref{de:strict} gives $U(\mathbf{p})>U(a_i)$. By linearity in
the action probabilities at $I_i$, reallocating the probability $\alpha$
from $a_i$ to $\mathbf{p}$ increases player $i$'s expected payoff by
$\alpha(U(\mathbf{p})-U(a_i))>0$, contradicting Nash equilibrium.
\end{proof}

\begin{proposition}
\label{prop:dominance-ne-extension}
In a finite perfect-recall game, every Nash equilibrium of a game obtained by successively removing strictly or weakly dominated actions extends to a Nash equilibrium of the original game, provided each action is dominated in the game from which it is removed.
\end{proposition}

\begin{proof}
Consider a single removal of action $a_i$ at information set $I_i$, and let $\mathbf{p}$ be a distribution over the remaining actions that dominates $a_i$. Extend a Nash equilibrium of the reduced game by assigning probability zero to $a_i$ and choosing arbitrary behavior at any information sets removed from the tree. This preserves the equilibrium payoffs. Any unilateral deviation by another player also gives the same payoff as in the reduced game, since player~$i$ still assigns probability zero to $a_i$.

Now consider a deviation by player~$i$ that assigns probability $\alpha$ to $a_i$. Reallocate this probability to the remaining actions according to $\mathbf{p}$, keeping all other choices fixed. If the deviation reaches $I_i$ with probability zero, this change leaves the payoff unchanged. Otherwise, dominance and linearity in the action probabilities at $I_i$ imply that the change weakly increases the payoff. The resulting deviation uses only remaining actions on every path it can reach, so its payoff is attainable in the reduced game. Since the reduced profile is a Nash equilibrium, neither deviation can be profitable. This proves the claim for one removal. Applying the argument in reverse order proves it for any sequence of removals.
\end{proof}

The example games considered were created using the open-source software package Gambit~\cite{Savani25a:Gambit}. Gambit has tools that allow the user to remove ``strictly dominated'' or ``strictly or weakly dominated'' actions, and these procedures can be repeated to iteratively remove multiple actions sequentially; however, the manual describes the dominance criteria informally and gives neither a formal definition nor details of the algorithms applied.

In the example game from Figure~\ref{fi:example}, Gambit correctly identifies that action 2 for player 2 at their top information set and action 1 for player 2 at their bottom information set are both strictly dominated, and removes these to construct the smaller game in Figure~\ref{fi:example-dom}. However, as for strong domination, Gambit fails to recognize that action 2 is strictly dominated by action 1 for player 1 in the reduced game. This example demonstrates that Gambit's procedure does not remove all dominated actions (though it does not necessarily imply that Gambit is only removing strongly dominated actions).

\section{Algorithm for identifying dominated actions in two-player games}
\label{se:algorithm}
Suppose we want to determine whether an action $c$ for player 1 is
dominated at a reachable information set $I$ in a two-player
perfect-recall extensive-form game $G$. We assume that player 1 has
no further decision on any path after taking any action at $I$.
Player 2 may still make decisions after $I$, and the players' actions
need not be publicly observable. Let $A(I)$ denote the set of actions
available at $I$. We consider whether $c$ is dominated by a probability
distribution over $A(I)\setminus\{c\}$.

Let $\mathbf{y}$ be a realization plan for player 2, and suppose that player 1
takes the actions on their sequence leading to $I$ with probability one.
Let $r_I(\mathbf{y})$ be the resulting probability of reaching $I$. For each
$a\in A(I)$, let $s_a$ denote player 1's sequence ending with action
$a$ at $I$. Since player 1 makes no further decision after $I$,
$(\mathbf{A}\mathbf{y})_{s_a}$ is the contribution to their expected payoff from paths
through $I$ when they choose $a$ there. Paths that do not pass through
$I$ contribute equally to the two strategies being compared.

Consider the following optimization problem:
\begin{equation}\label{eq:full}
\begin{aligned}
\max_{\mathbf{p},\epsilon}\quad & \epsilon \\
\text{s.t.}\quad
& \sum_{a\in A(I)\setminus\{c\}} p_a (\mathbf{A}\mathbf{y})_{s_a} - (\mathbf{A}\mathbf{y})_{s_c}
  \geq \epsilon\, r_I(\mathbf{y})
  && \text{for all } \mathbf{y}\geq \mathbf{0} \text{ with } \mathbf{F}\mathbf{y}=\mathbf{f}, \\
& \sum_{a\in A(I)\setminus\{c\}} p_a = 1, \\
& p_a\geq 0 && \text{for all } a\in A(I)\setminus\{c\}.
\end{aligned}
\end{equation}

\begin{proposition}\label{pr:strict-margin}
Action $c$ is strictly dominated at $I$ if and only if
Problem~\ref{eq:full} has a feasible solution with $\epsilon>0$.
\end{proposition}

\begin{proof}
For a fixed distribution $\mathbf{p}$, let $g_{\mathbf{p}}(\mathbf{y})$ denote the left-hand side
of the first constraint in Problem~\ref{eq:full}. By perfect recall,
player 1 has the same preceding sequence at every node in $I$.
If player 1 plays that sequence with probability $q$, then the
probability of reaching $I$ is $q\,r_I(\mathbf{y})$ and the payoff gain from
switching from $c$ to $\mathbf{p}$ at $I$ is $q\,g_{\mathbf{p}}(\mathbf{y})$. Thus, if the
constraint holds with $\epsilon>0$, the switch strictly increases
player 1's payoff whenever $I$ is reached with positive probability.

Conversely, suppose $\mathbf{p}$ strictly dominates $c$. For every pure
opponent plan $\mathbf{y}$ with $r_I(\mathbf{y})>0$, we have $g_{\mathbf{p}}(\mathbf{y})>0$ by
Definition~\ref{de:strict}. Since $I$ is reachable and there are
finitely many pure opponent plans, the minimum of $g_{\mathbf{p}}(\mathbf{y})/r_I(\mathbf{y})$
over these plans is some $\epsilon>0$. For every pure plan with
$r_I(\mathbf{y})=0$, we also have $g_{\mathbf{p}}(\mathbf{y})=0$. Hence
$g_{\mathbf{p}}(\mathbf{y})\geq\epsilon r_I(\mathbf{y})$ for every pure plan. Both sides are
linear in $\mathbf{y}$, and every opponent realization plan is a convex
combination of pure realization plans, so the inequality holds
for every feasible $\mathbf{y}$.
\end{proof}

Although Problem~\ref{eq:full} has only polynomially many variables,
its first constraint must hold for every realization plan $\mathbf{y}$ of
player~2. By linearity, it suffices to check the realization plans
corresponding to pure strategies, but there may be exponentially many
such plans. Thus, explicitly listing these constraints need not yield
a linear program of polynomial size in the game tree.
We now give a polynomial-size linear program for
Problem~\ref{eq:full}. Write $r_I(\mathbf{y})=\mathbf{d}_I^T \mathbf{y}$, where $\mathbf{d}_I$ is
indexed by the sequences of player 2. To construct $\mathbf{d}_I$, initialize
it to zero. For each node $h\in I$, add the product of the chance
probabilities on the path to $h$ to the entry for player 2's sequence
on that path. If player 2 has not yet acted, use the entry for the
empty sequence.

For each $a\in A(I)$, let $\mathbf{A}_{s_a,:}$ denote the row of $\mathbf{A}$
corresponding to sequence $s_a$, and define
\[
\mathbf{b}(\mathbf{p})=\sum_{a\in A(I)\setminus\{c\}} p_a \mathbf{A}_{s_a,:}^T
      - \mathbf{A}_{s_c,:}^T.
\]
For fixed $\mathbf{p}$ and $\epsilon$, the first constraint in
Problem~\ref{eq:full} is equivalent to
\[
\min_{\mathbf{y}\geq \mathbf{0}:\,\mathbf{F}\mathbf{y}=\mathbf{f}} \bigl(\mathbf{b}(\mathbf{p})-\epsilon \mathbf{d}_I\bigr)^T \mathbf{y} \geq 0.
\]
The feasible set of this minimization problem is the nonempty,
bounded realization-plan polytope of player 2. By linear programming
duality, Problem~\ref{eq:full} is therefore equivalent to
\begin{equation}\label{eq:strict-lp}
\begin{aligned}
\max_{\mathbf{p},\epsilon,\boldsymbol{\lambda}}\quad & \epsilon \\
\text{s.t.}\quad
& \mathbf{F}^T\boldsymbol{\lambda} \leq \mathbf{b}(\mathbf{p})-\epsilon \mathbf{d}_I, \\
& \mathbf{f}^T\boldsymbol{\lambda} \geq 0, \\
& \sum_{a\in A(I)\setminus\{c\}} p_a = 1, \\
& p_a \geq 0 \qquad \text{for all } a\in A(I)\setminus\{c\},
\end{aligned}
\end{equation}
where $\boldsymbol{\lambda}$ and $\epsilon$ are unrestricted in sign. The numbers
of variables and constraints are polynomial in the size of the game
tree. With rational payoffs and chance probabilities, all coefficients
can be computed in polynomial time. By Proposition~\ref{pr:strict-margin},
$c$ is strictly dominated if and only if this linear program admits a
feasible solution with $\epsilon>0$.

We next consider weak dominance. Let $\bar {\mathbf{y}}$ be the realization plan
of the strategy that chooses uniformly among the available actions at
each information set of player~2. Consider the linear program
\begin{equation}\label{eq:weak-lp}
\begin{aligned}
\max_{\mathbf{p},\boldsymbol{\lambda}}\quad & \mathbf{b}(\mathbf{p})^T\bar {\mathbf{y}} \\
\text{s.t.}\quad
& \mathbf{F}^T\boldsymbol{\lambda} \leq \mathbf{b}(\mathbf{p}), \\
& \mathbf{f}^T\boldsymbol{\lambda} \geq 0, \\
& \sum_{a\in A(I)\setminus\{c\}} p_a = 1, \\
& p_a \geq 0 \qquad \text{for all } a\in A(I)\setminus\{c\},
\end{aligned}
\end{equation}
where $\boldsymbol{\lambda}$ is unrestricted in sign. Action $c$ is weakly
dominated at $I$ if and only if Problem~\ref{eq:weak-lp} has a
feasible solution with positive objective value.

To see this, setting $\epsilon=0$ in the dual constraints of
Problem~\ref{eq:strict-lp} ensures that the payoff gain $\mathbf{b}(\mathbf{p})^T \mathbf{y}$
is nonnegative for every opponent realization plan $\mathbf{y}$. The plan
$\bar {\mathbf{y}}$ is a convex combination of the realization plans of all
pure opponent strategies with positive weight on each: equivalently,
player~2 can independently sample a uniform action at each information
set before play begins. Thus, once all gains are nonnegative,
$\mathbf{b}(\mathbf{p})^T\bar {\mathbf{y}}>0$ if and only if the gain is positive against at least
one pure opponent strategy. Any such strategy reaches $I$ with
positive probability when player~1 follows the sequence leading to
$I$, since otherwise its gain would be zero. For every strategy of
player~1 in Definition~\ref{de:weak}, the probability $q$ of following
that sequence is positive, so multiplying the gain by $q$ preserves
strictness. This gives the strict inequality required by the definition.

We can perform this procedure for every action at each reachable information set after which player~1 has no further decision on any path, and analogously for player~2. The procedure is summarized in Algorithm~\ref{al:dom}. Since the number of actions is linear in the size of the game tree, each scan involves solving a linear number of polynomial-size linear programs and therefore runs in polynomial time. We can repeat the procedure in the reduced game until no further action can be removed at any such information set. Each round that removes an action reduces the number of remaining actions, so the number of rounds is also polynomial. Thus, iterated removal of dominated actions at these information sets can be performed in polynomial time. The procedure applies to both zero-sum and general-sum two-player games with perfect recall under the stated last-decision assumption. It also identifies actions dominated by a probability distribution over other actions at the same information set, including cases in which no single other action dominates.

\begin{algorithm}[!ht]
\caption{Dominance test for action $c$ at a reachable information set $I$
after which player~1 has no further decision}
\label{al:dom}
\begin{algorithmic}[1]
\IF{$A(I)=\{c\}$}
    \RETURN action $c$ is not dominated
\ENDIF
\STATE Construct the sequence-form payoff matrix $\mathbf{A}$ for player~1
and the realization-plan constraints $\mathbf{F}\mathbf{y}=\mathbf{f}$, $\mathbf{y}\geq \mathbf{0}$, for player~2.
\STATE Compute the reach vector $\mathbf{d}_I$ and form $\mathbf{b}(\mathbf{p})$ from the
payoff rows for the actions at $I$.
\STATE Solve LP~\eqref{eq:strict-lp}.
\IF{the optimal objective value is positive}
    \RETURN action $c$ is strictly dominated
\ENDIF
\STATE Compute the realization plan $\bar {\mathbf{y}}$ for player~2's strategy
that chooses uniformly among the available actions at each information set.
\STATE Solve LP~\eqref{eq:weak-lp}.
\IF{the LP is feasible and its optimal objective value is positive}
    \RETURN action $c$ is weakly dominated
\ENDIF
\RETURN action $c$ is not dominated
\end{algorithmic}
\end{algorithm}

\begin{theorem}\label{th:strict-2}
In a two-player perfect-recall extensive-form game with rational payoffs
and chance probabilities, there exists a polynomial-time algorithm that
determines whether a given action is strictly dominated at a reachable
information set after which the acting player has no further decision
on any path.
\end{theorem}

\begin{theorem}\label{th:weak-2}
In a two-player perfect-recall extensive-form game with rational payoffs
and chance probabilities, there exists a polynomial-time algorithm that
determines whether a given action is weakly dominated at a reachable
information set after which the acting player has no further decision
on any path.
\end{theorem}

\begin{theorem}\label{th:iterated-2}
In a two-player perfect-recall extensive-form game with rational payoffs
and chance probabilities, iterated removal of strictly or weakly
dominated actions can be performed in polynomial time when each removal
is restricted to a reachable information set after which the acting
player has no further decision on any path in the current reduced game.
\end{theorem}

\section{Experiments}
\label{se:exp}
Now that we have defined dominated actions and shown that they can be identified efficiently under the assumptions of Section~\ref{se:algorithm}, we would like to investigate whether removing them can substantially reduce the size of a game in practice. Poker has been widely studied as a test domain for imperfect-information games. The most popular variant regularly played by humans is No-Limit Texas Hold'em (NLHE). Two-player NLHE works as follows. Initially two players each have a \emph{stack} of chips. One player, called the \emph{small blind}, initially puts $k$ worth of chips in the middle, while the other player, called the \emph{big blind}, puts in $2k$. The chips in the middle are known as the \emph{pot}, and will go to the winner of the hand. Next, there is an initial round of betting. The player whose turn it is to act can choose from three available options:
\begin{itemize}
\item \emph{Fold:} Give up on the hand, surrendering the pot to the opponent.
\item \emph{Call:} Put in the minimum number of chips needed to match the number of chips put into the pot by the opponent. For example, if the opponent has put in \$1000 and we have put in \$400, a call would require putting in \$600 more. A call of zero chips is also known as a \emph{check}.
\item \emph{Bet:} Put in additional chips beyond what is needed to call. A bet can be of any size from 1 chip up to the number of chips a player has left in their stack, provided it exceeds some minimum value and is a multiple of the smallest chip denomination. A bet of all of one's remaining chips is called an \emph{all-in} bet (aka a \emph{shove}). 
\end{itemize}

The initial round of betting ends if a player has folded, if there has been a bet and a call, or if both players have checked. If the round ends without a player folding, then three public cards are revealed face-up on the table (called the \emph{flop}) and a second round of betting takes place. Then one more public card is dealt (called the \emph{turn}) and a third round of betting, followed by a fifth public card (called the \emph{river}) and a final round of betting. If a player ever folds, the other player wins all the chips in the pot. If the final betting round is completed without a player folding (or if a player is all-in at an earlier round), then both players reveal their private cards, and the player with the best five-card hand (out of their two private cards and the five public cards) wins the pot (it is divided equally if there is a tie).

In some situations the blinds are very large relative to stack sizes. This can happen frequently at later stages in poker \emph{tournaments}, where the blinds increase after a certain time duration. A common rule is that when stack sizes are below around 8 big blinds a \emph{shove-or-fold} strategy should be employed where each player only goes all-in or folds~\cite{Mersenneary11:Correctly}. Study of optimal shove-or-fold strategies has been considered for 2-player~\cite{Miltersen07:Near-optimal} and 3-player~\cite{Ganzfried08:Computing,Ganzfried09:Computing} poker tournament endgames. The poker site Americas Cardroom\footnote{\url{https://www.americascardroom.eu/}} has specific ``All-in or Fold'' tables with up to 4 players where players are only allowed to play shove-or-fold strategies. The initial stack sizes at these tables are either 8 or 10 times the big blind; the highest stake available has blinds of \$100 and \$200 with initial stacks of \$2000. 

In the two-player NLHE shove-or-fold game, each player has 169 strategically distinct hands (13 pocket pairs and $\frac{13\cdot 12}{2} = 78$ combinations of each non-paired offsuit and suited hand). Let player 1 denote the small blind and player 2 denote the big blind. Player 1 can shove or fold, and if player 1 shoves, player 2 can call or fold. Since each player acts at most once on any path, the algorithm in Section~\ref{se:algorithm} applies at every decision in this game. We assume that the blinds are small blind $k = 100$ and big blind $2k = 200$, and initial stacks are 1600 (8 times the big blind). For the poker experiments, we test dominance separately for each hand over the opponent strategies consistent with the actions remaining in the game. We first remove all strictly dominated actions for player 1, followed by removing strictly dominated actions for player 2 (note that removing weakly dominated actions does not provide additional benefit in this game). It turns out that 85 actions for player 1 are removed and 96 actions for player 2 are removed. Thus, the initial game with 169 hands per player can be reduced to a game where player 1 must make a decision with 84 hands and player 2 with 73 hands; the number of decision points has been reduced by over 50\%. It turns out that performing an additional round of removing dominated actions does not remove any further actions.\footnote{We note that no actions are strongly dominated for either player in this game (i.e., Candidate Definition~\ref{cd:strict-strong}). For example, consider the action of player 1 folding with pocket aces (the best starting hand). The payoff of this action is -\$100. If player 1 goes all-in, there are leaf nodes where player 1 will receive a payoff of -\$1600.}

\begin{table}[!th]
\centering
\caption{Non-dominated actions for player 1 in two-player NLHE allin-fold with 5 big blind stacks. Suited hands are in the upper right and unsuited hands are in the lower left.}
\label{ta:5bb-p1}
\begin{tabular}{c||*{13}{c|}}
  &A &K &Q &J &T &9 &8 &7 &6 &5 &4 &3 &2 \\ \hline \hline
A &S(1) &S(1) &S(1) &S(1) &S(1) &S(1) &S(1) &S(1) &S(1) &S(1) &S(1) &S(1) &S(1)\\ \hline
K &S(1) &S(1) &S(1) &S(1) &S(1) &S(1) &S(1) &S(1) &S(1) &S(1) &S(1) &S(1) &S(1)\\ \hline
Q &S(1) &S(1) &S(1) &S(1) &S(1) &S(1) &S(1) &S(1) &S(1) &S(1) &S(1) &S(1) &S(1)\\ \hline
J &S(1) &S(1) &S(1) &S(1) &S(1) &S(1) &S(1) &S(1) &S(1) &S(1) &S(1) &S(1) &S(1)\\ \hline
T &S(1) &S(1) &S(1) &S(1) &S(1) &S(1) &S(1) &S(1) &S(1) &S(1) &S(1) &? &?\\ \hline
9 &S(1) &S(1) &S(1) &S(1) &S(1) &S(1) &S(1) &S(1) &S(1) &? &? &? &?\\ \hline
8 &S(1) &S(1) &S(1) &S(1) &S(1) &S(1) &S(1) &S(1) &S(2) &? &? &? &?\\ \hline
7 &S(1) &S(1) &S(1) &S(1) &S(1) &? &? &S(1) &S(3) &? &? &? &?\\ \hline
6 &S(1) &S(1) &S(1) &S(1) &? &? &? &? &S(1) &? &? &? &?\\ \hline
5 &S(1) &S(1) &S(1) &S(1) &? &? &? &? &? &S(1) &? &? &?\\ \hline
4 &S(1) &S(1) &S(1) &S(3) &? &F(4) &? &? &? &? &S(1) &? &?\\ \hline
3 &S(1) &S(1) &S(1) &? &? &F(4) &F(2) &F(4) &F(4) &? &? &S(1) &?\\ \hline
2 &S(1) &S(1) &S(1) &? &F(4) &F(2) &F(2) &F(2) &F(2) &F(4) &F(2) &F(2) &S(1)\\ \hline
\end{tabular}
\end{table}

\begin{table}[!ht]
\centering
\caption{Non-dominated actions for player 2 in two-player NLHE allin-fold with 5 big blind stacks. Suited hands are in the upper right and unsuited hands are in the lower left.}
\label{ta:5bb-p2}
\begin{tabular}{c||*{13}{c|}}
  &A &K &Q &J &T &9 &8 &7 &6 &5 &4 &3 &2 \\ \hline \hline
A &S(1) &S(1) &S(1) &S(1) &S(1) &S(1) &S(1) &S(1) &S(1) &S(1) &S(1) &S(1) &S(1)\\ \hline
K &S(1) &S(1) &S(1) &S(1) &S(1) &S(1) &S(1) &S(1) &S(1) &S(1) &S(1) &S(1) &S(1)\\ \hline
Q &S(1) &S(1) &S(1) &S(1) &S(1) &S(1) &S(1) &S(1) &S(1) &S(1) &S(1) &S(1) &S(1)\\ \hline
J &S(1) &S(1) &S(1) &S(1) &S(1) &S(1) &S(1) &S(1) &S(1) &S(1) &? &? &?\\ \hline
T &S(1) &S(1) &S(1) &S(1) &S(1) &S(1) &S(1) &S(1) &? &? &? &? &?\\ \hline
9 &S(1) &S(1) &S(1) &S(1) &S(1) &S(1) &S(1) &S(1) &? &? &? &? &?\\ \hline
8 &S(1) &S(1) &S(1) &S(1) &S(3) &? &S(1) &S(1) &? &? &? &F(4) &F(2)\\ \hline
7 &S(1) &S(1) &S(1) &S(3) &? &? &? &S(1) &? &? &? &F(2) &F(1)\\ \hline
6 &S(1) &S(1) &S(1) &? &? &? &? &? &S(1) &? &? &F(2) &F(1)\\ \hline
5 &S(1) &S(1) &S(1) &? &? &? &F(4) &F(2) &F(2) &S(1) &? &F(2) &F(1)\\ \hline
4 &S(1) &S(1) &S(1) &? &? &F(2) &F(1) &F(1) &F(1) &F(1) &S(1) &F(1) &F(1)\\ \hline
3 &S(1) &S(1) &? &? &F(4) &F(2) &F(1) &F(1) &F(1) &F(1) &F(1) &S(1) &F(1)\\ \hline
2 &S(1) &S(1) &? &? &F(4) &F(1) &F(1) &F(1) &F(1) &F(1) &F(1) &F(1) &S(1)\\ \hline
\end{tabular}
\end{table}

Next we consider the setting where the stacks are 5 times the big blind (i.e., 1000). In this game four rounds of iterated removal of dominated actions are needed. The first round removes 108 dominated actions for player 1 and 115 for player 2; the second round removes 8 for player 1 and 8 for player 2; the third round removes 2 for player 1 and 2 for player 2; the fourth round removes 6 for player 1 and 4 for player 2. In the final reduced game player 1 must make a decision with only 45 hands while player 2 must make a decision with only 40 hands. Tables~\ref{ta:5bb-p1} and~\ref{ta:5bb-p2} show the non-dominated actions for each player with parentheses indicating the iteration at which the alternative action was removed. `S' indicates shove for player 1 (or call for player 2), `F' indicates fold, and `?' indicates that neither action was dominated. With stacks of 4 big blinds, 13 hands for player 1 and 10 hands for player 2 remain undecided after 4 rounds of removing dominated actions. With stacks of 3 big blinds, only 2 hands for player 1 and 1 hand for player 2 remain undecided after 3 rounds. These results demonstrate that iteratively removing dominated actions can significantly reduce the size of realistic games. While for simplicity we considered two-player zero-sum games, for which the full game can be solved directly by a linear program, the computational benefit for games with more than two players could be much more significant.

\section{Conclusion}
\label{se:conc}
Dominance is a fundamental concept in game theory. It is well-understood in normal-form games, but its impact in imperfect-information games has so far been largely unexplored. We consider several plausible definitions of dominated actions which we demonstrate to be problematic; however, one of them which we denote as \emph{strong domination} can be useful as an efficient preprocessing step. We present a new definition that addresses these limitations. In two-player games with perfect recall, we show that both strictly and weakly dominated actions can be identified in polynomial time at any information set after which the acting player has no further decision on any path. Our algorithms identify actions that are dominated by a probability distribution over other actions at the same information set, even when actions are not publicly observable. Iterative removal at information sets satisfying this condition can also be performed in polynomial time. We demonstrate that removing dominated actions can play a significant role in reducing the size of realistic imperfect-information games. This can serve as an efficient preprocessing step before computation of a Nash equilibrium. Proposition~\ref{prop:dominance-ne-extension} justifies this preprocessing step: any Nash equilibrium computed in the reduced game extends to a Nash equilibrium of the original game.

Subsequent recent work has shown that removal of dominated actions enabled an algorithm to compute a Nash equilibrium in a three-player imperfect-information game in under 3 seconds while the algorithm was unable to solve the full game in 24 hours~\cite{Ganzfried26:Quadratic}. In practice our algorithms could be sped up by several heuristics such as selection of the order of traversal of information sets and players. Recent work has shown that some games contain many ``mistake'' actions that are played with probability zero in all Nash equilibria but are not dominated~\cite{Ganzfried19:Mistakes}. Thus, there is potentially more future ground to explore on efficient game reduction by elimination of poor actions. The most immediate directions for future study are to determine the complexity of finding dominated actions when the acting player may have further decisions after the information set in question, when the game has more than two players, or when perfect recall does not hold.

\bibliographystyle{plain}
\bibliography{C://FromBackup/Research/refs/dairefs}

\end{document}